\documentclass{article}

\PassOptionsToPackage{square,numbers}{natbib}
\usepackage[preprint]{nips_2018}

\usepackage[utf8]{inputenc} 
\usepackage[T1]{fontenc}    
\usepackage{hyperref}       
\usepackage{url}            
\usepackage{booktabs}       
\usepackage{amsfonts}       
\usepackage{nicefrac}       
\usepackage{microtype}      
\usepackage{natbib}

\usepackage[english]{babel}
\usepackage{amsmath}
\usepackage{amsthm}
\usepackage{amssymb}
\usepackage{graphicx}
\usepackage{algorithm,algpseudocode}
\usepackage{pgfplots}
\usepackage{float}
\usepackage{verbatim}
\usepackage{enumitem}
\pgfplotsset{compat=1.14}
\newtheorem{thm}{Theorem}

\DeclareMathOperator{\normRed}{\oplus}

\title{Online normalizer calculation for softmax}
\author{
	Maxim Milakov\\
    NVIDIA\\
    \texttt{mmilakov@nvidia.com}
    \And
    Natalia Gimelshein\\
    NVIDIA\\
    \texttt{ngimelshein@nvidia.com}
}
\date{April 2018}

\begin{document}
\maketitle

\begin{abstract}

The Softmax function is ubiquitous in machine learning, multiple previous works suggested faster alternatives for it. In this paper we propose a way to compute classical Softmax with fewer memory accesses and hypothesize that this reduction in memory accesses should improve Softmax performance on actual hardware. The benchmarks confirm this hypothesis: Softmax accelerates by up to $1.3$x and Softmax+TopK combined and fused by up to $5$x. 

\end{abstract}

\section{Introduction}

Neural networks models are widely used for language modeling, for tasks such as machine translation \cite{2014arXiv1409.3215S} and speech recognition \cite{6296526}. These models compute word probabilities taking into account the already generated part of the sequence. The probabilities are usually computed by a Projection layer, which "projects" hidden representation into the output vocabulary space, and a following Softmax function, which transforms raw logits into the the vector of probabilities. Softmax is utilized not only for neural networks, for example, it is employed in multinomial logistic regression \cite{PolytomousLR}.

A number of previous works suggested faster alternatives to compute word probabilities. Differentiated Softmax \cite{2015arXiv151204906C} and SVD-Softmax \cite{NIPS2017_7130} replace the projection layer - which is usually just a matrix multiplication - with more computationally efficient alternatives. Multiple variants of Hierarchical Softmax \cite{conf_icassp_Goodman01, 5947611, 2016arXiv160904309G} split a single Projection+Softmax pair into multiple much smaller versions of these two functions organized in tree-like structures. Sampled-based approximations, such as Importance Sampling \cite{Bengio+Senecal-2003}, Noise Contrastive Estimation \cite{2012arXiv1206.6426M}, and Blackout \cite{2015arXiv151106909J} accelerate training by running Softmax on select elements of the original vector. Finally, Self-Normalized Softmax \cite{Devlin14fastand} augments the objective function to make the softmax normalization term close to $1$ (and skip computing it during inference).

This is not an exhaustive list, but, hopefully, a representative one. Almost all of the approaches still need to run the original Softmax function, either on full vector or reduced one. There are two exceptions that don't need to compute the softmax normalization term: training with Noise Contrastive Estimation and inference with Self-Normalized Softmax. All others will benefit from the original Softmax running faster.

To the best of our knowledge there has been no targeted efforts to improve the performance of the original Softmax function. We tried to address this shortcoming and figured out a way to compute Softmax with fewer memory accesses. We benchmarked it to see if those reductions in memory accesses translate into performance improvements on a real hardware.

\section{Original softmax}
Function \(y=Softmax(x)\) is defined as:
\begin{equation}
\label{eq:naive_softmax}
y_i=\frac{e^{x_i}}{\sum\limits_{j=1}^{V}{e^{x_j}}}
\end{equation}
where \(x, y \in \mathbb{R}^V\). The naive implementation (see algorithm~\ref{alg:naive_softmax}) scans the input vector two times - one to calculate the normalization term $d_V$ and another to compute output values $y_i$ - effectively doing three memory accesses per vector element: two loads and one store.
\begin{algorithm}[ht]
\caption{Naive softmax}
\label{alg:naive_softmax}
\begin{algorithmic}[1]
\State $d_0\gets 0$
\For{$j\gets 1, V$}
\State $d_j\gets d_{j-1}+e^{x_j}$ \label{overflow_line}
\EndFor
\For{$i\gets 1, V$}
\State $y_i\gets \frac{e^{x_i}}{d_V}$
\EndFor
\end{algorithmic}
\end{algorithm}

Unfortunately, on real hardware, where the range of numbers represented is limited, the line~\ref{overflow_line}  of the algorithm~\ref{alg:naive_softmax} can overflow or underflow due to the exponent. There is a safe form of \eqref{eq:naive_softmax}, which is immune to this problem:
\begin{equation}
\label{eq:safe_softmax}
y_i=\frac{e^{x_i-\max\limits_{k=1}^{V}{x_k}}}{\sum\limits_{j=1}^{V}{e^{x_j-\max\limits_{k=1}^{V}{x_k}}}}
\end{equation}
\begin{algorithm}[ht]
\caption{Safe softmax}
\label{alg:safe_softmax}
\begin{algorithmic}[1]
\State $m_0\gets -\infty$
\For{$k\gets 1, V$}
\State $m_k\gets \max(m_{k-1},{x_k})$
\EndFor
\State $d_0\gets 0$
\For{$j\gets 1, V$}
\State $d_j\gets d_{j-1}+e^{x_j-m_V}$
\EndFor
\For{$i\gets 1, V$}
\State $y_i\gets \frac{e^{x_i-m_V}}{d_V}$
\EndFor
\end{algorithmic}
\end{algorithm}
All major DL frameworks are using this safe version for the Softmax computation: TensorFlow \cite{tensorflow2015-whitepaper} v1.7, PyTorch \cite{paszke2017automatic} (with Caffe2) v0.4.0, MXNET \cite{mxnet_learningsys2015} v1.1.0, Microsoft Cognitive Toolkit \cite{Seide:2016:CMO:2939672.2945397} v2.5.1, and Chainer \cite{chainer_learningsys2015} v5.0.0a1. But Safe Softmax does three passes over input vector: The first one calculates the maximum value $m_V$, the second one - normalization term $d_V$, and the third one - final values $y_i$, see algorithm~\ref{alg:safe_softmax}; This results in 4 memory access per vector element overall. We want to improve on that.

\section{Online normalizer calculation}
The algorithm~\ref{alg:online_softmax} calculates both the maximum value $m$ and the normalization term $d$ in a single pass over input vector with negligible additional cost of two operations per vector element. It reduces memory accesses from 4 down to 3 per vector element for the Softmax function evaluation. Inspiration came from the numerically stable variance calculation online algorithm, see \cite{doi:10.1080/00401706.1962.10490022}.
\begin{algorithm}[ht]
\caption{Safe softmax with online normalizer calculation}
\label{alg:online_softmax}
\begin{algorithmic}[1]
\State $m_0\gets -\infty$\label{alg:start_thm_line}
\State $d_0\gets 0$
\For{$j\gets 1, V$}
\State $m_j\gets \max\left(m_{j-1},x_j\right)$\label{alg:m_line}
\State $d_j\gets d_{j-1}\times e^{m_{j-1}-m_j} +e^{x_j-m_j}$\label{alg:d_line}
\EndFor\label{alg:end_thm_line}
\For{$i\gets 1, V$}
\State $y_i\gets \frac{e^{x_i-m_V}}{d_V}$
\EndFor
\end{algorithmic}
\end{algorithm}

Essentially, the algorithm keeps the maximum value $m$ and the normalization term $d$ as it iterates over elements of the input array. At each iteration it needs to adjust the normalizer $d$ to the new maximum $m_j$ and only then add new value to the normalizer.

\begin{thm}
\label{thm:fast}
The lines~\ref{alg:start_thm_line}-\ref{alg:end_thm_line} of the algorithm~\ref{alg:online_softmax} compute $m_V=\max\limits_{k=1}^{V}{x_k}$ and $d_V=\sum_{j=1}^{V}{e^{x_j-m_V}}$ 
\end{thm}

\begin{proof}
We will use a proof by induction.
\begin{itemize}[label=$\lozenge$]
\item \emph{Base case}: $V=1$
\begin{flalign*}
m_1\gets &x_1 && \text{by line~\ref{alg:m_line} of the algorithm~\ref{alg:online_softmax}}\\
	=&\max\limits_{k=1}^{1}{x_k}\\
d_1\gets &e^{x_1-m_1} && \text{by line~\ref{alg:d_line} of the algorithm~\ref{alg:online_softmax}}\\
	=&\sum\nolimits_{j=1}^{1}{e^{x_j-m_1}}
\end{flalign*}
The theorem holds for $V=1$.
\item \emph{Inductive step}: We assume the theorem statement holds for $V=S-1$, that is the lines~\ref{alg:start_thm_line}-\ref{alg:end_thm_line} of the algorithm~\ref{alg:online_softmax} compute $m_{S-1}=\max\limits_{k=1}^{S-1}{x_k}$ and $d_{S-1}=\sum_{j=1}^{S-1}{e^{x_j-m_{S-1}}}$. Let's see what the algorithm computes for $V=S$\\
\begin{flalign*}
m_S\gets &\max\left(m_{S-1},x_S\right) && \text{by line~\ref{alg:m_line} of the algorithm~\ref{alg:online_softmax}}\\
	=&\max(\max\limits_{k=1}^{S-1}{x_k},x_S) && \text{by the inductive hypothesis}\\
    =&\max\limits_{k=1}^{S}{x_k} \\
d_S\gets &d_{S-1}\times e^{m_{S-1}-m_S} +e^{x_S-m_S} && \text{by line~\ref{alg:d_line} of the algorithm~\ref{alg:online_softmax}}\\
	=&\left(\sum\nolimits_{j=1}^{S-1}{e^{x_j-m_{S-1}}}\right)\times e^{m_{S-1}-m_S}+e^{x_S-m_S} && \text{by the inductive hypothesis}\\
	=&\sum\nolimits_{j=1}^{S-1}{e^{x_j-m_S}}+e^{x_S-m_S}\\
	=&\sum\nolimits_{j=1}^S{e^{x_j-m_S}}
\end{flalign*}
The inductive step holds as well.\qedhere
\end{itemize}
\end{proof}

The algorithm~\ref{alg:online_softmax} is proved to compute the Softmax function as defined in \eqref{eq:safe_softmax}. It is also safe:
\begin{itemize}
\item $m_j$ is the running maximum, $m_j \in \left[\min\limits_{k=1}^{V}{m_k},\max\limits_{k=1}^{V}{m_k}\right], \forall j\in 1,V$; $m_j$ cannot underflow or overflow.
\item $d_j$ is also bounded: $1 \leq d_j \leq j, \forall j\in 1,V$. It can be easily proven by induction. The 32-bit floating point storage for $d_j$ guarantees processing of up to $1.7*10^{37}$ elements in vector $x$ without overflow. It is a reasonably large amount, but if your vector is even larger you need to use the 64-bit floating point storage for $d_j$.
\end{itemize}

The algorithm~\ref{alg:safe_softmax} provides the same guarantees: $1 \leq d_j \leq j, \forall j\in 1,V$. 

In the remainder of this paper we will call algorithm~\ref{alg:online_softmax} "Online Softmax".

\subsection{Parallel online normalizer calculation}

The lines~\ref{alg:start_thm_line}-\ref{alg:end_thm_line} of the algorithm~\ref{alg:online_softmax} define a sequential way of calculating the normalization term in a single pass over input vector. Modern computing devices allow running multiple threads concurrently; We need to have a parallel version of the algorithm to fully utilize devices. We define a generalized version of the online normalizer calculation:

\begin{equation}
\label{eq:generalized_online}
\begin{bmatrix}m_V\\d_V\end{bmatrix} = \begin{bmatrix}x_1\\1\end{bmatrix}\normRed\begin{bmatrix}x_2\\1\end{bmatrix}\normRed ...\normRed\begin{bmatrix}x_V\\1\end{bmatrix}
\end{equation}

where \(x_i, m_V, d_V \in \mathbb{R}\). The binary operation $\normRed:\mathbb{R}^2\times \mathbb{R}^2\rightarrow\mathbb{R}^2$ is defined as:

\begin{equation}
\begin{bmatrix}m_i\\d_i\end{bmatrix}\normRed\begin{bmatrix}m_j\\d_j\end{bmatrix} = \begin{bmatrix}\max\left(m_i,m_j\right)\\d_i\times e^{m_i-\max\left(m_i,m_j\right)}+d_j\times e^{m_j-\max\left(m_i,m_j\right)}\end{bmatrix}
\end{equation}

Applying \eqref{eq:generalized_online} sequentially from left to right is equivalent to running lines~\ref{alg:start_thm_line}-\ref{alg:end_thm_line} of the algorithm~\ref{alg:online_softmax}. The operation $\normRed$ is associative, which enables parallel evaluation of \eqref{eq:generalized_online}. It is also commutative, which provides the flexibility needed to make parallel implementations more efficient. We omit the proofs for these two statements for brevity.

\section{Softmax and top-k fusion}

Online Softmax (algorithm~\ref{alg:online_softmax}) does three memory accesses per vector element: one load for the normalizer calculation, one load and one store for computing Softmax function values $y_i$. Inference with the beam search for auto-regressive models has TopK following Softmax, and this TopK doesn't need to compute all $y_i$ values. This enables even bigger improvements.

The TopK function is producing the vector of K integer indices referencing the largest values in the input vector, along with those values:
\begin{equation}
TopK\left(y\right)=(v,z): v_i=y_{z_i}, v_i\geq y_j, \forall i \in \left[1,K\right], \forall j \notin z
\end{equation}
where \(y \in \mathbb{R}^V, z \in \mathbb{Z}^K, v\in \mathbb{R}^K\).

\begin{algorithm}[ht]
\caption{Online softmax and top-k}
\label{alg:online_softmax_topk}
\begin{algorithmic}[1]
\State $m_0\gets -\infty$
\State $d_0\gets 0$
\State $u\gets \left\{-\infty,-\infty,\dots,-\infty\right\}^T, u\in\mathbb{R}^{K+1}$ \Comment The 1st $K$ elems will hold running TopK values
\State $p\gets \left\{-1,-1,\dots,-1\right\}^T, p\in \mathbb{Z}^{K+1}$ \Comment ... and their indices
\For{$j\gets 1, V$}
\State $m_j\gets \max\left(m_{j-1},x_j\right)$
\State $d_j\gets d_{j-1}\times e^{m_{j-1}-m_j} +e^{x_j-m_j}$
\State $u_{K+1}\gets x_j$ \Comment Initialize $K+1$ elem with new value from input vector
\State $p_{K+1}\gets j$ \Comment ... and its index
\State $k\gets K$\label{alg:start_partial_topk_lines} \Comment Sort $u$ in descending order, permuting $p$ accordingly. The first K elements are already sorted, so we need just a single loop, inserting the last element in the correct position.
\While{$k\geq 1 \text{ and } u_k<u_{k+1}$}
	\State swap$\left(u_k,u_{k+1}\right)$
    \State swap$\left(p_k,p_{k+1}\right)$
    \State $k\gets k-1$
\EndWhile\label{alg:end_partial_topk_lines}
\EndFor
\For{$i\gets 1, K$} \Comment The algorithm stores only K values and their indices
\State $v_i\gets \frac{e^{u_i-m_V}}{d_V}$
\State $z_i\gets p_i$
\EndFor
\end{algorithmic}
\end{algorithm}

The TopK needs to load each element of the input vector at least once. Running Safe Softmax and the TopK separately requires 5 accesses per input element and 4 accesses if we use Online Softmax instead of Safe Softmax (but still run them separately, one after another). If we improve on the algorithm~\ref{alg:online_softmax} and keep not only running values of $m$ and $d$ (when iterating over the input vector), but also the vectors of TopK input values $u$ and their indices $p$ - as in the algorithm~\ref{alg:online_softmax_topk} - we can run this Softmax+TopK fusion with just one memory access per element of the input vector.

\section{Benchmarking}
Online normalizer calculation reduces the number of memory accesses for the Softmax and Softmax+TopK functions. The softmax function has a very low flops per byte ratio; that means the memory bandwidth should be limiting the performance, even for Online Softmax with its additional few floating point operations per element. Fewer memory accesses should translate into performance improvements, and experiments confirm this.

We implemented a benchmark for GPUs using CUDA C. The benchmark utilizes \href{https://nvlabs.github.io/cub/}{CUB} v1.8.0 for fast parallel reductions. All experiments were run on NVIDIA Tesla V100 PCIe 16 GB, ECC on, persistent mode on, CUDA Toolkit 9.1. Source code of the benchmark is available at \href{https://github.com/NVIDIA/online-softmax}{github.com/NVIDIA/online-softmax}.

\subsection{Benchmarking softmax}

We benchmarked all 3 Softmax algorithms - Naive, Safe, and Online - on different vector sizes for the batch sizes of 4,000 and 10. The large batch case corresponds to the training or batch inference with enough input vectors to saturate the device and and the small batch case corresponds to online inference with too few vectors to occupy the device fully.

\begin{figure}[ht]
\centering
\begin{tikzpicture}
\pgfplotsset{set layers}
\begin{semilogxaxis} [
scale only axis,
ylabel = Performance improvement,
xmin=35,
xmax=500000,
ymin=0.65,
ymax=2,
axis x line=none,
axis y line*=right,
ybar,
bar width=1.4,
ymajorgrids = true,
grid style = dotted,
legend pos = south east,
]
\addplot[color=teal!60!white,fill=teal!30!white] table [x=V, y expr=\thisrow{OnlineSoftmax} / \thisrow{SafeSoftmax}] {benchmark_softmax_4000vectors.txt};
\addlegendentry {Online/Safe}
\end{semilogxaxis}
\begin{semilogxaxis} [
axis on top=true,
scale only axis,
scaled ticks=false,
xlabel = Vector size $V$,
ylabel = Elements per second,
ymin=0,
xmin=35,
xmax=500000,
axis x line*=bottom,
axis y line*=left,
legend pos = south west,
xmajorgrids = true,
grid style = dotted,
legend cell align={right},
cycle list={
  magenta!80!black,every mark/.append style={fill=magenta},mark=triangle*\\
  teal!80!black,every mark/.append style={fill=teal},mark=square*\\
  olive!80!black,every mark/.append style={fill=olive},mark=*\\
},
]
\addplot table [x=V, y=NaiveSoftmax] {benchmark_softmax_4000vectors.txt};
\addlegendentry {Naive}
\addplot table [x=V, y=OnlineSoftmax] {benchmark_softmax_4000vectors.txt};
\addlegendentry {Online}
\addplot table [x=V, y=SafeSoftmax] {benchmark_softmax_4000vectors.txt};
\addlegendentry {Safe}
\end{semilogxaxis}
\end{tikzpicture}
\caption{Benchmarking softmax, Tesla V100, fp32, batch size 4000 vectors}
\label{fig:softmax_4000}
\end{figure}
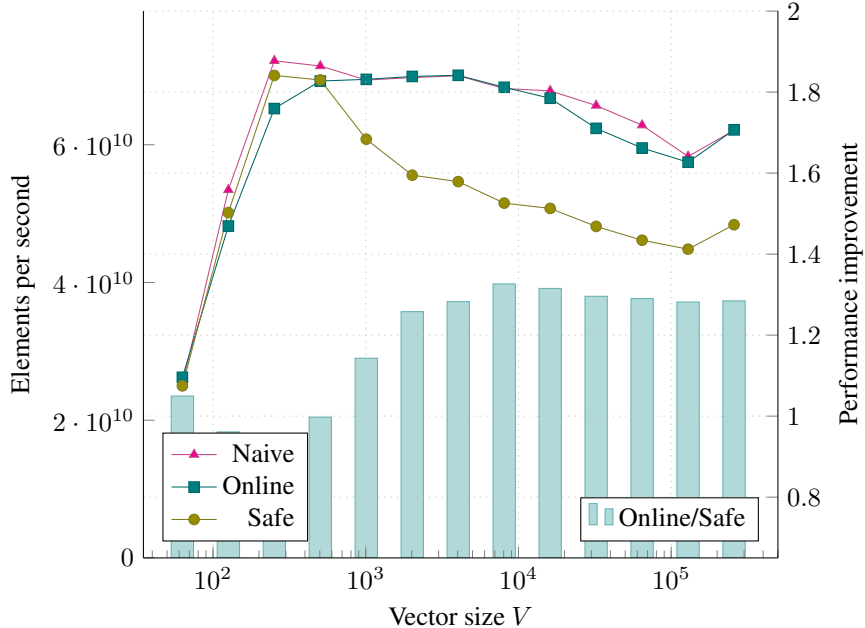

For the large batch case (see figure~\ref{fig:softmax_4000}) all three algorithms perform similarly up until $V=1000$ vector size. The NVIDIA Visual Profiler shows that at that point L1 and L2 cache thrashing starts to make all three algorithms limited by the DRAM bandwidth. When this happens Online and Naive algorithms are getting faster than Safe one, quickly achieving $\sim 1.3$x at $V=4000$ (look for bars in the chart, they are showing performance improvement of Online Softmax over Safe Softmax). This is quite close to $1.33$x reduction in memory accesses for those algorithms.

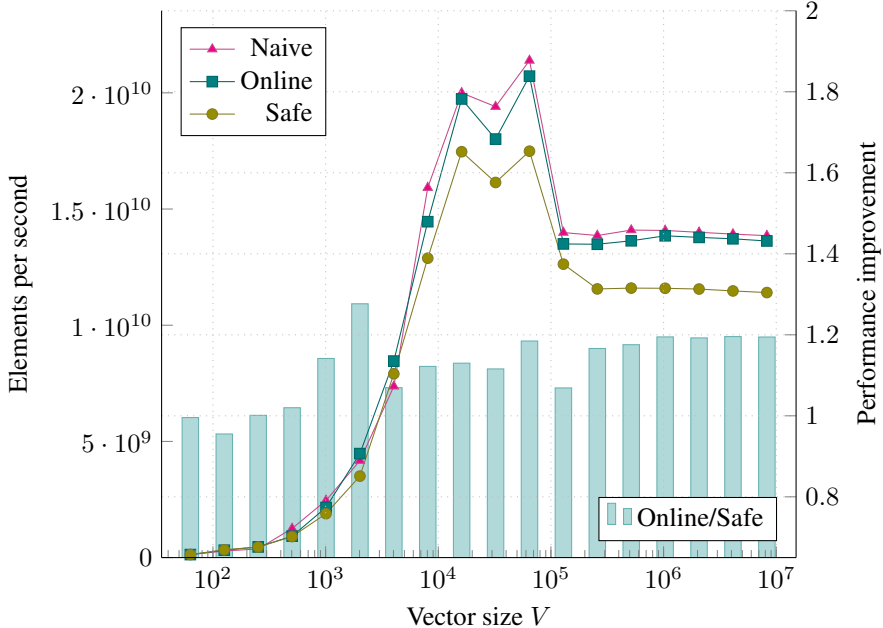
\begin{figure}[ht]
\centering
\begin{tikzpicture}
\pgfplotsset{set layers}
\begin{semilogxaxis} [
scale only axis,
ylabel = Performance improvement,
xmin=35,
xmax=15000000,
ymin=0.65,
ymax=2,
axis x line=none,
axis y line*=right,
ybar,
bar width=1.4,
ymajorgrids = true,
grid style = dotted,
legend pos = south east,
]
\addplot[color=teal!60!white,fill=teal!30!white] table [x=V, y expr=\thisrow{OnlineSoftmax} / \thisrow{SafeSoftmax}] {benchmark_softmax_10vectors.txt};
\addlegendentry {Online/Safe}
\end{semilogxaxis}
\begin{semilogxaxis} [
axis on top=true,
scale only axis,
scaled ticks=false,
xlabel = Vector size $V$,
ylabel = Elements per second,
ymin=0,
xmin=35,
xmax=15000000,
axis x line*=bottom,
axis y line*=left,
legend pos = north west,
xmajorgrids = true,
grid style = dotted,
legend cell align={right},
cycle list={
  magenta!80!black,every mark/.append style={fill=magenta},mark=triangle*\\
  teal!80!black,every mark/.append style={fill=teal},mark=square*\\
  olive!80!black,every mark/.append style={fill=olive},mark=*\\
},
]
\addplot table [x=V, y=NaiveSoftmax] {benchmark_softmax_10vectors.txt};
\addlegendentry {Naive}
\addplot table [x=V, y=OnlineSoftmax] {benchmark_softmax_10vectors.txt};
\addlegendentry {Online}
\addplot table [x=V, y=SafeSoftmax] {benchmark_softmax_10vectors.txt};
\addlegendentry {Safe}
\end{semilogxaxis}
\end{tikzpicture}
\caption{Benchmarking softmax, Tesla V100, fp32, batch size 10 vectors}
\label{fig:softmax_10}
\end{figure}

The absolute performance for small batch case is lower for all algorithms, see figure~\ref{fig:softmax_10}. The benchmark is running one threadblock per vector; thus small batch case - with 10 vectors - has just 10 threadblocks in the grid. This is not enough to saturate the GPU, both compute and the memory subsystem are underutilized, various latencies are exposed. As in the batch inference case, all three algorithms show similar performance up to $V=1000$ vector size. After that Naive and Online algorithms outperform Safe one by $\sim 1.15$x.

\subsection{Benchmarking softmax and top-k}

We benchmarked Safe Softmax followed by the TopK (running one after another), Safe Softmax fused with the TopK into a single function, and Online Softmax fused with TopK, again, for 2 cases: 4,000 and 10 vectors. We picked up $K=5$ in TopK for all runs.

\begin{figure}[ht]
\centering
\begin{tikzpicture}
\pgfplotsset{set layers}
\begin{semilogxaxis} [
scale only axis,
ylabel = Performance improvement,
xmin=35,
xmax=500000,
ymin=0,
ymax=6,
axis x line=none,
axis y line*=right,
ybar,
bar width=1.4,
ymajorgrids = true,
grid style = dotted,
legend pos = south east,
]
\addplot[color=teal!60!white,fill=teal!30!white] table [x=V, y expr=\thisrow{OnlineSoftmaxFusedTopK} / \thisrow{SafeSoftmaxUnfusedTopK}] {benchmark_softmax_4000vectors.txt};
\addlegendentry {Online fused/Safe unfused}
\end{semilogxaxis}
\begin{semilogxaxis} [
axis on top=true,
scale only axis,
scaled ticks=false,
xlabel = Vector size $V$,
ylabel = Elements per second,
ymin=0,
xmin=35,
xmax=500000,
axis x line*=bottom,
axis y line*=left,
legend pos = north west,
xmajorgrids = true,
grid style = dotted,
legend cell align={right},
cycle list={
  teal!80!black,every mark/.append style={fill=teal},mark=square*\\
  olive!80!black,every mark/.append style={fill=olive},mark=*\\
  magenta!80!black,every mark/.append style={fill=magenta},mark=triangle*\\
},
]
\addplot table [x=V, y=OnlineSoftmaxFusedTopK] {benchmark_softmax_4000vectors.txt};
\addlegendentry {Online Softmax + TopK fused}
\addplot table [x=V, y=SafeSoftmaxFusedTopK] {benchmark_softmax_4000vectors.txt};
\addlegendentry {Safe Softmax + TopK fused}
\addplot table [x=V, y=SafeSoftmaxUnfusedTopK] {benchmark_softmax_4000vectors.txt};
\addlegendentry {Safe Softmax + TopK unfused}
\end{semilogxaxis}
\end{tikzpicture}
\caption{Benchmarking softmax and top-k, Tesla V100, fp32, batch size 4000 vectors}
\label{fig:softmax_topk_4000}
\end{figure}
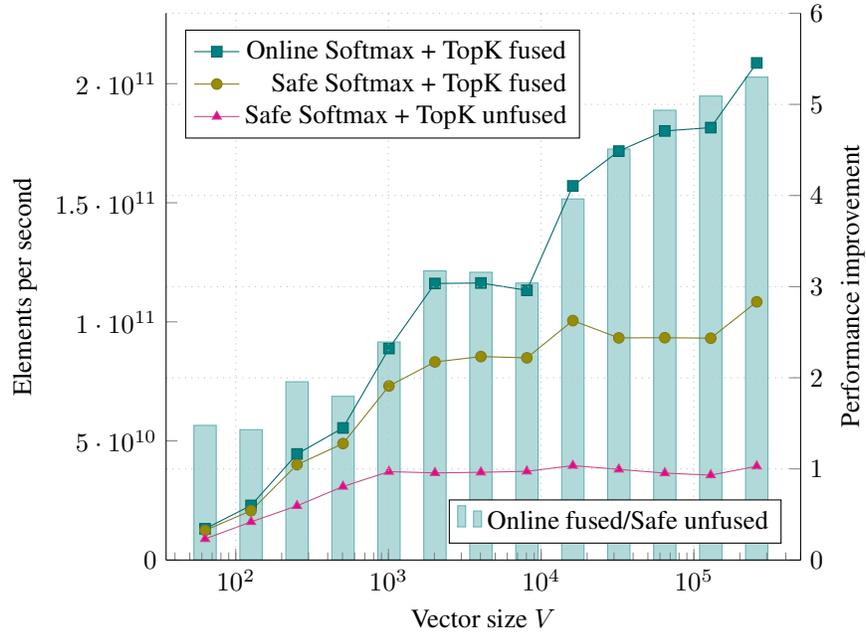

Online fused version is running considerably faster than Safe unfused one. For large batch case - see figure~\ref{fig:softmax_topk_4000} - the performance improvement starts at $1.5$x and goes up as vector size $V$ increases approaching $5$x at $V=25000$, which corresponds to $5$x reduction in memory accesses. This $5$x comes from $2.5$x due to function fusion and $2$x due to Online Softmax itself.

\begin{figure}[ht]
\centering
\begin{tikzpicture}
\pgfplotsset{set layers}
\begin{semilogxaxis} [
scale only axis,
ylabel = Performance improvement,
xmin=35,
xmax=15000000,
ymin=0,
ymax=4,
axis x line=none,
axis y line*=right,
ybar,
bar width=1.4,
ymajorgrids = true,
grid style = dotted,
legend pos = south east,
]
\addplot[color=teal!60!white,fill=teal!30!white] table [x=V, y expr=\thisrow{OnlineSoftmaxFusedTopK} / \thisrow{SafeSoftmaxUnfusedTopK}] {benchmark_softmax_10vectors.txt};
\addlegendentry {Online fused/Safe unfused}
\end{semilogxaxis}
\begin{semilogxaxis} [
axis on top=true,
scale only axis,
scaled ticks=false,
xlabel = Vector size $V$,
ylabel = Elements per second,
ymin=0,
xmin=35,
xmax=15000000,
axis x line*=bottom,
axis y line*=left,
legend pos = north west,
xmajorgrids = true,
grid style = dotted,
legend cell align={right},
legend style={cells={align=right}},
cycle list={
  teal!80!black,every mark/.append style={fill=teal},mark=square*\\
  olive!80!black,every mark/.append style={fill=olive},mark=*\\
  magenta!80!black,every mark/.append style={fill=magenta},mark=triangle*\\
},
]
\addplot table [x=V, y=OnlineSoftmaxFusedTopK] {benchmark_softmax_10vectors.txt};
\addlegendentry {Online Softmax +\\ TopK fused}
\addplot table [x=V, y=SafeSoftmaxFusedTopK] {benchmark_softmax_10vectors.txt};
\addlegendentry {Safe Softmax +\\ TopK fused}
\addplot table [x=V, y=SafeSoftmaxUnfusedTopK] {benchmark_softmax_10vectors.txt};
\addlegendentry {Safe Softmax +\\ TopK unfused}
\end{semilogxaxis}
\end{tikzpicture}
\caption{Benchmarking softmax and top-k, Tesla V100, fp32, batch size 10 vectors}
\label{fig:softmax_topk_10}
\end{figure}
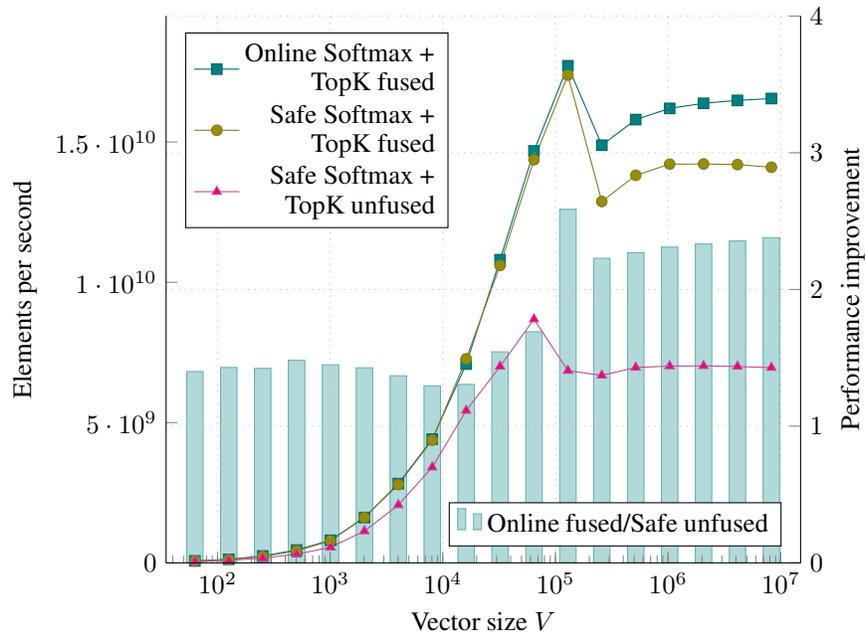

In the small batch case (see figure~\ref{fig:softmax_topk_10}) Online fused version outperforms Safe unfused one by $1.5$x-$2.5$x. It cannot achieve $5$x because the GPU is underutilized and the performance is limited not by the memory bandwidth, but by various latencies. Yet the reduction in memory accesses helps even in this latency limited case. In small batch case fusion only already brings substantial performance improvements, switching to Online Softmax helps improve performance even further.

The benchmark shows these levels of performance improvement for relatively small $K$ only. The cost of keeping partial TopK results - as in the lines~\ref{alg:start_partial_topk_lines}-\ref{alg:end_partial_topk_lines} of the algorithm~\ref{alg:online_softmax_topk} - increases quickly as $K$ gets bigger: the performance improvement drops to $3.5$x for $K=10$, $2$x for $K=15$, $1.4$x for $K=30$, and degrades further for bigger $K$s. For these cases the TopK is dominating (in terms of runtime) over the Softmax. Getting rid of separate Softmax and fusing the normalization term calculation into the TopK is still beneficial, but the value goes down as TopK is taking more and more time.

\section{Results}

We introduced the way to calculate the normalizer for the Softmax function in a single pass over input data, which reduces memory accesses by $1.33$x for the Softmax function alone. Benchmarks on Tesla V100 show that this materializes in $1.15$x performance improvements for $V\geq 1000$ vector sizes, and for the large batch mode it goes up to $1.3$x when $V\geq 4000$.

If one is using Naive Softmax then switching to Online version improves numerical accuracy with no performance hit or a negligible one.

When the TopK follows the Softmax the new single-pass normalizer calculation enables efficient fusion of these 2 functions resulting in $5$x fewer memory accesses for Softmax+TopK combined. We observed $1.5$x-$5$x performance improvement on Tesla V100, with this $5$x improvement coming from $2.5$x with fusion and $2$x with Online Softmax itself.

These performance improvements could be applied not only to the classical Softmax function; They are orthogonal to many other Softmax optimization techniques including Hierarchical Softmax, Importance Sampling, and SVD-Softmax.

\section{Discussion}

Online Softmax is running up to $1.3$x faster on the latest generation GPU than the one used by major DL frameworks. It also enables very efficient fusion of the Softmax with following TopK showing up to $5$x performance improvement over the traditional Safe Softmax and TopK running separately.

Could we see significantly different speed-ups or even slow-downs on different compute devices, for example CPUs? We didn't do experiments for those, but if the original code is vectorized and one manages to keep it vectorized for the online normalizer (and partial TopK) calculation then similar speedups could probably be expected.

There could be a way to improve the performance further. The resulting Softmax and even Softmax+TopK fused are still limited by the memory bandwidth, so fusing them with the preceding layer will avoid memory round trip, thus improving performance. This change is more challenging though.

\subsubsection*{Acknowledgments}

We would like to thank Christoph Angerer for his valuable comments and suggestions.

\bibliographystyle{hunsrtnat}
\bibliography{main}

\end{document}